\documentclass[runningheads]{llncs}
\usepackage{lineno,hyperref}
\usepackage{url}
\usepackage{stmaryrd}
\usepackage{tikz-cd}
\usepackage{yfonts}
\usepackage[utf8]{inputenc}
\usepackage[T1]{fontenc}
\usepackage [autostyle]{csquotes}
\usepackage{caption}
\usepackage{mathrsfs}
\usepackage{amsfonts}
\usepackage{eqnarray}
\usepackage{mathtools}

\usepackage{float}
\usepackage{amssymb}
\usepackage{amsmath}
\usepackage{enumerate}
\usepackage{lipsum}
\usepackage{graphicx}
\newtheorem{thm}{Theorem}
\newtheorem{lem}{Lemma}
\newtheorem{prop}{Proposition}
\newtheorem{cor}{Corollary}
\newtheorem{defn}{Definition}

\newtheorem{rem}{Remark}

\begin{document}
		\title{A Lindstr\"om Theorem for Fitting's Modal Logic over a
			Finite Heyting Algebra}
			\titlerunning{A Lindstr\"om Theorem for Fitting's Modal Logic..}
				\author{Litan Kumar Das\inst{1}\orcidID{0000-0002-0315-0974}}
			\authorrunning{Litan Kumar Das}
			%
			\institute{Department of Mathematics, Mahishadal Raj College, Mahishadal,
				721628, West Bengal, India\\
				\email{ld06iitkgp@gmail.com}}
		\maketitle
		\begin{abstract}
			We establish a Lindstr\"om-style maximality theorem for
			Maruyama's exact-truth-test presentation of Fitting's modal logic
			over a fixed finite Heyting algebra and crisp Kripke frames.
			Unlike the existing characterization over finite MTL-chains, no
			linearity or distinguished coatom is assumed. Exact truth tests
			yield Boolean tests for designated and non-designated values and
			a derived existential modality sufficient for the saturation
			argument. We prove that every abstract extension which is compact,
			has the Tarski Union Property, and is strongly invariant under
			bisimulation is $1$-expressively equivalent to Maruyama's version of Fitting's Heyting-valued modal logic. As a consequence, every exact-value fibre of an extension formula
			is definable in Maruyama's exact-truth-test modal language.
		\end{abstract}
	%
	%
	%
	%
	%
	%
	%
		
\section{Introduction}
\label{sec:introduction}
Many-valued modal logics combine modal reasoning with evaluations
taking values in non-Boolean algebras. Fitting's early work
developed Kripke semantics for modal languages whose formulas take
values in a fixed finite algebra~\cite{fitting1991many}. Maruyama later
studied an algebraic presentation over a finite distributive
lattice $\mathcal L$, hence a finite Heyting algebra, and established
completeness with respect to $\mathcal L$-valued Kripke semantics
together with a finite model property~\cite{maruyama2009algebraic}.
A distinctive feature of this presentation is the use of unary
exact truth tests $T_\ell$, $\ell\in\mathcal L$, in place of
primitive truth constants for all non-Boolean values. Thus
$T_\ell\varphi$ has value $1$ precisely when $\varphi$ has value
$\ell$. Subsequent duality results provide, among other
consequences, compactness for the corresponding finite-valued
logics~\cite{maruyama2011dualities}.\\
Hennessy--Milner phenomena for many-valued modal logics over crisp
Kripke frames were studied by Marti and
Metcalfe~\cite{marti2014hennessy}. Their work relates modal
equivalence and bisimulation under suitable chain-based algebraic
conditions. In the present paper, a saturated-model
Hennessy--Milner argument is used instead over a finite Heyting
algebra.
Lindstr\"om-type theorems characterize the expressive strength of
a logic through model-theoretic properties such as compactness,
bisimulation invariance, and the Tarski Union Property; see, for
example, the classical modal result of de
Rijke~\cite{deRijke1995}. Lindstr\"om-type maximality results have also been developed in non-classical settings, including intuitionistic propositional logic~\cite{BadiaOlkhovikovIPL2020} and graded first-order
logic~\cite{BadiaNoguera2021}.

 More directly relevant here, Badia
and Olkhovikov proved a modal Lindstr\"om theorem over a fixed
finite MTL-chain with truth constants~\cite{badia2020lindstrom}.
However, no corresponding Lindstr\"om maximality result is
available for Maruyama's exact-truth-test presentation of
Fitting's modal logic when the fixed finite Heyting algebra is not
assumed to be linearly ordered. The finite-chain argument of
Badia and Olkhovikov uses the unique immediate predecessor of $1$
to separate designated from non-designated values, together with
linearity in the maximality argument. Hence that proof does not
extend directly to the present setting. Related exact-value and
frame-definability results for more general finite algebras
\cite{badia2023frame} address a different problem and do
not yield maximality for arbitrary abstract extensions.\\
The purpose of the present paper is to show that the exact truth
tests available in the Maruyama's version of Fitting's Heyting-valued modal language provide the
separation mechanism required to remove the chain assumption from
the Lindstr\"om argument. For an arbitrary formula $\varphi$, put $D\varphi:=T_1\varphi,\; N\varphi:=T_0(T_1\varphi)$. Both formulas are Boolean-valued: $D\varphi=1\;\Longleftrightarrow\;
\varphi=1$, whereas $N\varphi=1\;\Longleftrightarrow\;\varphi\neq1$. The same truth tests also define the derived modality $\Diamond_1\varphi :=T_0\bigl(\Box T_0(T_1\varphi)\bigr)$, which expresses the existence of an accessible state at which
$\varphi$ has designated value $1$. The novelty is the use of this exact-test
separation mechanism to replace the chain-dependent coatom
argument in exact-value recovery, saturation, and the
Lindstr\"om separation step.

Our main theorem states that, for every fixed finite Heyting
algebra $\mathcal L$, the Maruyama's version of Fitting's Heyting-valued modal logic is
maximal, with respect to $1$-expressive power, among its abstract
extensions satisfying compactness, the Tarski Union Property, and
strong bisimulation invariance. The decisive replacement for the
finite-chain argument is the Boolean decision pair $T_1\psi\;\text{and}\; T_0(T_1\psi)$, which avoids the use of an immediate predecessor of $1$ or any comparability assumption on the elements of $\mathcal L$.\\
The paper is organized as follows.
Section~\ref{sec:base} introduces the Fitting--Maruyama base
logic and records the model-theoretic properties required later.
Section~\ref{sec:abstract} introduces abstract extensions,
exact-value recovery, successor types, and saturation.
Section~\ref{sec:saturated} constructs saturated elementary
extensions using Compactness and the Tarski Union Property.
Section~\ref{sec:lindstrom} establishes the Lindstr\"om maximality
theorem and derives exact-value definability as a consequence.
The final section concludes the paper and discusses directions for further work.

\section{Maruyama's version of Fitting's modal logic}
\label{sec:base}
We first recall the fragment of Maruyama's version of Fitting's
$\mathcal L$-valued modal logic and its crisp Kripke semantics
needed in the sequel; see
\cite{fitting1991many,maruyama2009algebraic,maruyama2011dualities}.
Throughout the paper, $\mathcal L=(L,\wedge,\vee,\to,0,1)$ denotes a fixed finite Heyting algebra. No linearity assumption is made on its lattice order. We write $\mathsf{FML}_{\mathcal L}$ for Maruyama's version of Fitting's $\mathcal L$-valued modal logic \cite{maruyama2009algebraic}. Let $\tau$ be a set of propositional variables. The formulas of $\mathsf{FML}_{\mathcal L}$ are generated by
\[
\varphi ::=p\mid0\mid1
\mid(\varphi\wedge\varphi)
\mid(\varphi\vee\varphi)
\mid(\varphi\to\varphi)
\mid T_\ell\varphi
\mid\Box\varphi,
\qquad \ell\in L.
\]
For each $\ell\in L$, the exact truth test $T_\ell$ is interpreted
on $\mathcal L$ by
\[
T_\ell(a)=
\begin{cases}
	1,&a=\ell,\\
	0,&a\neq\ell.
\end{cases}
\]
An $\mathcal L$-valued Kripke model for $\mathsf{FML}_{\mathcal L}$ is a triple $M=(W,R,V)$,
where $W\neq\varnothing$, $R\subseteq W\times W$ is a crisp
relation and $V:\tau\longrightarrow L^W$. The propositional connectives are evaluated pointwise in
$\mathcal L$, while
$\llbracket T_\ell\varphi\rrbracket_M(x)
=
T_\ell\bigl(\llbracket\varphi\rrbracket_M(x)\bigr)$ and
$\llbracket\Box\varphi\rrbracket_M(x)
=
\bigwedge\{
\llbracket\varphi\rrbracket_M(y):xRy
\}$, where $\bigwedge\varnothing=1$. We write $M,x\models\varphi$ when $\llbracket\varphi\rrbracket_M(x)=1$.
\subsection{Boolean tests and a derived existential modality}
For the Lindstr\"om argument, we introduce the following two
Boolean tests:
for every formula $\varphi$ define $D\varphi:=T_1\varphi, N\varphi:=T_0(T_1\varphi)$.
\begin{lem}
	\label{lem:boolean}
	For every pointed model $(M,x)$,
	\[
	\llbracket D\varphi\rrbracket_M(x)
	=
	\begin{cases}
		1,&\llbracket\varphi\rrbracket_M(x)=1,\\
		0,&\llbracket\varphi\rrbracket_M(x)\neq1,
	\end{cases}
	\]
	and
	\[
	\llbracket N\varphi\rrbracket_M(x)
	=
	\begin{cases}
		0,&\llbracket\varphi\rrbracket_M(x)=1,\\
		1,&\llbracket\varphi\rrbracket_M(x)\neq1.
	\end{cases}
	\]
	Consequently, $D\varphi$ and $N\varphi$ are $\{0,1\}$-valued, $D\varphi\vee N\varphi$ is valid, and
	$D\varphi\wedge N\varphi$ has value $0$ everywhere.
\end{lem}
\begin{proof}
	This follows immediately from the interpretations of $T_1$ and
	$T_0$.
\end{proof}
We further define the following derived existential modality: $\Diamond_1\varphi
:=
T_0\bigl(\Box N\varphi\bigr)
=
T_0\bigl(\Box T_0(T_1\varphi)\bigr)$.
\begin{lem}
	\label{lem:diamond}
	For every model $M$ and state $x$, $M,x\models\Diamond_1\varphi$ if and only if there exists $y$ such that $xRy
	\;\text{and}\;
	M,y\models\varphi$.
\end{lem}

\begin{proof}
	The values of $N\varphi$ belong to $\{0,1\}$, and
	$N\varphi$ has value $0$ exactly at states at which $\varphi$
	has value $1$. Hence $\llbracket\Box N\varphi\rrbracket_M(x)=0$ exactly when some $R$-successor of $x$ gives $N\varphi$ value $0$. Applying $T_0$ proves the assertion.
\end{proof}
\subsection{Basic model-theoretic properties}
We use the standard notion of bisimulation for many-valued
Kripke models with crisp accessibility; cf.\
\cite{marti2014hennessy,badia2020lindstrom}.
A bisimulation between $M=(W_M,R_M,V_M)\;\text{and}\;K=(W_K,R_K,V_K)$ is a relation $Z\subseteq W_M\times W_K$ such that, whenever $xZy$,
\begin{enumerate}[(i)]
	\item
	$V_M(p)(x)=V_K(p)(y)$ for every $p\in\tau$;
	\item
	if $xR_Mx'$, then some $y'\in W_K$ satisfies
	$yR_Ky'$ and $x'Zy'$;
	\item
	if $yR_Ky'$, then some $x'\in W_M$ satisfies
	$xR_Mx'$ and $x'Zy'$.
\end{enumerate}
We write $(M,x)\sim(K,y)$ when a bisimulation relates $x$ and
$y$.
\begin{prop}
	\label{prop:base-bisim}
	The logic $\mathsf{FML}_{\mathcal L}$ is strongly invariant under
	bisimulation: if $(M,x)\sim(K,y)$, then $\llbracket\varphi\rrbracket_M(x)
	=
	\llbracket\varphi\rrbracket_K(y)$ for every $\varphi\in\mathsf{FML}_{\mathcal L}$.
\end{prop}
\begin{proof}
	The proof is by induction on $\varphi$. The propositional and
	$T_\ell$ cases are immediate. For $\Box\psi$, the forth and back
	conditions, together with the induction hypothesis, show that $\{
	\llbracket\psi\rrbracket_M(x'):xRx'
	\}
	=
	\{
	\llbracket\psi\rrbracket_N(y'):yR'y'
	\}$. Taking meets yields the desired equality.
\end{proof}
Compactness of $\mathsf{FML}_{\mathcal L}$ follows from
Maruyama's duality theory~\cite{maruyama2011dualities}. 
\begin{prop}[Compactness]
	\label{prop:base-compact}
	If every finite subset of $\Gamma\subseteq\mathsf{FML}_{\mathcal L}(\tau)$ is satisfiable at value $1$, then $\Gamma$ is satisfiable at value $1$.
\end{prop}
We next verify the Tarski Union Property in its standard
elementary-chain form; cf.~\cite{badia2020lindstrom,BadiaNoguera2021}.
For models $M\subseteq K$, write $M\preccurlyeq_{\mathcal L}K$ if, for every $x\in W_M$ and every $\varphi\in\mathsf{FML}_{\mathcal L}(\tau)$, $\llbracket\varphi\rrbracket_M(x)
=
\llbracket\varphi\rrbracket_K(x)$. For an increasing chain $M_n=(W_{M_n},R_{M_n},V_{M_n})\;(n<\omega)$, its union is the model $M=(W_M,R_M,V_M)$, where $W_M=\bigcup_{n<\omega}W_{M_n},\;R_M=\bigcup_{n<\omega}R_{M_n}$, and, for $x\in W_{M_n}$, $V_M(p)(x)=V_{M_n}(p)(x)$. This is well defined: if
$x\in W_{M_n}\cap W_{M_m}$, say $n\leq m$, then
$M_n\subseteq M_m$, and hence $V_{M_n}(p)(x)=V_{M_m}(p)(x)$.
\begin{prop}[Tarski Union Property]
	\label{prop:base-tup}
	Suppose $M_0\preccurlyeq_{\mathcal L}M_1\preccurlyeq_{\mathcal L}\cdots$, is a countable elementary chain and put $M=\bigcup_{n<\omega}M_n$. Then $M_n\preccurlyeq_{\mathcal L}M
	\;(n<\omega)$.
	\end{prop}
\begin{proof}
	We prove by induction on $\varphi$ that, for every $i<\omega$
	and every $x\in W_{M_i}$, $\llbracket\varphi\rrbracket_{M_i}(x)
	=\llbracket\varphi\rrbracket_M(x)$.
	The non-modal cases are immediate. Let $\varphi=\Box\psi$ and put
	$a=\llbracket\Box\psi\rrbracket_{M_i}(x)$. Since $R_{M_i}[x]\subseteq R_M[x]$, the induction hypothesis gives $\llbracket\Box\psi\rrbracket_M(x)\leq a$. Conversely, let $xR_My$. Since $R_M=\bigcup_{n<\omega}R_{M_n}$, choose $j\ge i$ such that $xR_{M_j}y$. As $M_i\preccurlyeq_{\mathcal L}M_j$, $a=
	\llbracket\Box\psi\rrbracket_{M_j}(x)
	\le
	\llbracket\psi\rrbracket_{M_j}(y)$. By the induction hypothesis,
	$\llbracket\psi\rrbracket_{M_j}(y)
	=
	\llbracket\psi\rrbracket_M(y)$.
	Thus $a$ is a lower bound of $\{\llbracket\psi\rrbracket_M(y):xR_My\}$, and hence $a\leq\llbracket\Box\psi\rrbracket_M(x)$. Therefore
	$\llbracket\Box\psi\rrbracket_{M_i}(x)
	=
	\llbracket\Box\psi\rrbracket_M(x)$. Hence $M_i\preccurlyeq_{\mathcal L}M$ for every $i<\omega$.

\end{proof}

\section{Abstract extensions and saturation}
\label{sec:abstract}
To formulate expressive maximality, we adapt the abstract-extension
framework of~\cite{badia2020lindstrom,BadiaNoguera2021} to
$\mathsf{FML}_{\mathcal L}$, with all logics interpreted over the
same class of $\mathcal L$-valued crisp Kripke models.

\begin{defn}
	An \emph{abstract extension} $\mathscr L$ of
	$\mathsf{FML}_{\mathcal L}$ assigns to each signature $\tau$ a
	set $\mathscr L(\tau)$ of formulas satisfying the following
	conditions.
		\begin{enumerate}[(i)]
		\item
		If $\tau\subseteq\tau'$, then
		$\mathscr L(\tau)\subseteq\mathscr L(\tau')$;
		
		\item
		Every $\varphi\in\mathscr L(\tau)$ has a finite occurrence
		signature $\tau_\varphi\subseteq\tau$ such that $\varphi\in\mathscr L(\tau_\varphi)$;
		
		\item $\mathsf{FML}_{\mathcal L}(\tau)\subseteq\mathscr L(\tau)$,
		and $\mathscr L(\tau)$ is closed under $\wedge,\vee,\to,T_\ell\;(\ell\in L),\Box,0,1$;
		\item
	Every $\varphi\in\mathscr L(\tau)$ is assigned a value $\llbracket\varphi\rrbracket_M^{\mathscr L}(x)\in L$ at every pointed $\tau$-model, and  the connectives of $\mathsf{FML}_{\mathcal L}$ retain their usual semantics on arbitrary
	$\mathscr L$-formulas. In particular,
$	\llbracket T_\ell\varphi\rrbracket_M^{\mathscr L}(x)
	=
	T_\ell\!\left(
	\llbracket\varphi\rrbracket_M^{\mathscr L}(x)
	\right)$
	and
$	\llbracket\Box\varphi\rrbracket_M^{\mathscr L}(x)
	=\bigwedge_{xRy}
	\llbracket\varphi\rrbracket_M^{\mathscr L}(y)$;
		\item The interpretation is invariant under isomorphisms of pointed
		models: if $(M,x)$ and $(K,y)$ are isomorphic pointed
		$\tau$-models, then
	$	\llbracket\varphi\rrbracket_M^{\mathscr L}(x)
		=
		\llbracket\varphi\rrbracket_K^{\mathscr L}(y)$, for every $\varphi\in\mathscr L(\tau)$;
		\item The interpretation is invariant under expansion and reduct of
		signatures: if $\tau\subseteq\tau'$, $\varphi\in\mathscr L(\tau)$,
		and $M'=(W',R',V')$ is a $\tau'$-model, then, for every $x\in W'$,
		$\llbracket\varphi\rrbracket_{M'}^{\mathscr L}(x)
		=
		\llbracket\varphi\rrbracket_{M'\upharpoonright\tau}^{\mathscr L}(x)$, where $M'\upharpoonright\tau=(W',R',V'\upharpoonright\tau)$ denotes the $\tau$-reduct of $M'$.
		
	\end{enumerate}
\end{defn}
We write $M,x\models_{\mathscr L}\varphi$ when
$\llbracket\varphi\rrbracket_M^{\mathscr L}(x)=1$, and define
$\operatorname{Th}_{\mathscr L}(M,x)
=
\{
\varphi\in\mathscr L(\tau):
M,x\models_{\mathscr L}\varphi
\}$. For $\Gamma\subseteq\mathscr L(\tau)$, write $M,x\models_{\mathscr L}\Gamma
\iff
M,x\models_{\mathscr L}\gamma
\text{ for every }\gamma\in\Gamma$. For $M\subseteq K$, write $M\preccurlyeq_{\mathscr L}K$ if
$\operatorname{Th}_{\mathscr L}(M,x)
=
\operatorname{Th}_{\mathscr L}(K,x) \; (x\in M)$.
Compactness and the Tarski Union Property for $\mathscr L$ are
defined as in Propositions~\ref{prop:base-compact} and ~\ref{prop:base-tup}, with $\mathscr L$ in place of
$\mathsf{FML}_{\mathcal L}$.
Following Badia and Olkhovikov~\cite{badia2020lindstrom}, we say
that an abstract extension $\mathscr L$ is \emph{strongly invariant under bisimulation} if
$(M,x)\sim(K,y)$ implies
$\llbracket\varphi\rrbracket_M^{\mathscr L}(x)
=
\llbracket\varphi\rrbracket_K^{\mathscr L}(y)
$ for every $\varphi\in\mathscr L(\tau)$.
Since $\mathscr L(\tau)$ is closed under $T_\ell$ and $\Box$, the abbreviations
$D\varphi=T_1\varphi,\;
N\varphi=T_0(T_1\varphi),\;
\Diamond_1\varphi=T_0(\Box N\varphi)$ are available for every $\varphi\in\mathscr L(\tau)$. Their usual semantics ensures that Lemmas~\ref{lem:boolean} and ~\ref{lem:diamond} remain valid.\\
Two formulas over the same signature $\tau$ are
$1$-equivalent, written $\varphi\equiv_1\psi$, if they have
value $1$ at exactly the same pointed $\tau$-models. We write
$\mathscr L_1\leq_1\mathscr L_2$ if, for every signature $\tau$
and every $\varphi\in\mathscr L_1(\tau)$, there is
$\psi\in\mathscr L_2(\tau)$ such that $\varphi\equiv_1\psi$,
and $\mathscr L_1\equiv_1\mathscr L_2$ if both expressive inclusions hold.

\begin{lem}[Exact-value recovery]
	\label{lem:recovery}
	For pointed models $(M,x)$ and $(K,y)$,
$	\operatorname{Th}_{\mathscr L}(M,x)
	=
	\operatorname{Th}_{\mathscr L}(K,y)$
	if and only if
	$\llbracket\varphi\rrbracket_M^{\mathscr L}(x)
	=
	\llbracket\varphi\rrbracket_K^{\mathscr L}(y)$ for every $\varphi\in\mathscr L(\tau)$.
\end{lem}

\begin{proof}
	Only the forward implication requires proof. Suppose
	$\llbracket\varphi\rrbracket_M^{\mathscr L}(x)=\ell$.
	Then $M,x\models_{\mathscr L}T_\ell\varphi$. Equality of theories yields $K,y\models_{\mathscr L}T_\ell\varphi$, and therefore $\llbracket\varphi\rrbracket_K^{\mathscr L}(y)=\ell$.
	
\end{proof}
We use the successor-type notion underlying the saturation
argument of~\cite{badia2020lindstrom}, formulated here using the
derived modality $\Diamond_1$.
\begin{defn}
	Let $\mathscr L$ be an abstract extension and let $M=(W_M,R_M,V_M)$ be a
	$\tau$-model. A set $\Gamma\subseteq\mathscr L(\tau)$ is an
	\emph{$\mathscr L$-successor type} at $(M,x)$ if, for every finite
	$\Gamma_0\subseteq\Gamma$, there exists $y\in W_M$ such that
	$xR_My$ and $M,y\models_{\mathscr L}\Gamma_0$. The model $M$ is \emph{$\mathscr L$-saturated} if every
	$\mathscr L$-successor type $\Gamma$ at $(M,x)$ is realized by
	some $y\in W_M$, i.e., $xR_My\;\text{and}\;M,y\models_{\mathscr L}\Gamma$.

\end{defn}
By Lemma~\ref{lem:diamond}, $\Gamma$ is a successor type at $x$ iff
$M,x\models_{\mathscr L}
\Diamond_1\left(\bigwedge\Gamma_0\right)$
for every finite $\Gamma_0\subseteq\Gamma$, where
$\bigwedge\varnothing=1$.
\begin{prop}[Hennessy--Milner property]
	\label{prop:hm}
	Let $M=(W_M,R_M,V_M)$ and $K=(W_K,R_K,V_K)$ be
	$\mathsf{FML}_{\mathcal L}$-saturated models. If
	$\llbracket\varphi\rrbracket_M(x)
	=
	\llbracket\varphi\rrbracket_K(y)$
	for every $\varphi\in\mathsf{FML}_{\mathcal L}(\tau)$, then
	$(M,x)\sim(K,y)$.
\end{prop}
\begin{proof}
	Let $\operatorname{Th}_{\mathrm{FML}}(M,x)
	:=
	\{\varphi\in\mathsf{FML}_{\mathcal L}(\tau):
	M,x\models\varphi\}.$
	Define a relation $Z\subseteq W_M\times W_K$ by
	$uZv
	\;\Longleftrightarrow\;
	\operatorname{Th}_{\mathrm{FML}}(M,u)
	=
	\operatorname{Th}_{\mathrm{FML}}(K,v)$.
	By Lemma \ref{lem:recovery}, $uZv$ implies $\llbracket\varphi\rrbracket_M(u)
	=
	\llbracket\varphi\rrbracket_K(v)$ for every $\varphi\in\mathsf{FML}_{\mathcal L}(\tau)$,
	and hence atomic agreement. Suppose $uZv$ and $uR_Mu'$. Put
	$\Gamma=\operatorname{Th}_{\mathrm{FML}}(M,u')$.
	For every finite $\Gamma_0\subseteq\Gamma$,
	$M,u\models
	\Diamond_1\left(\bigwedge\Gamma_0\right)$. Since $uZv$, we have $K,v\models
	\Diamond_1\Bigl(\bigwedge\Gamma_0\Bigr)$. Hence $\Gamma$ is a successor type at $(K,v)$. Saturation gives
	$vR_Kv'$ with
	\[
	\Gamma
	\subseteq
	\operatorname{Th}_{\mathrm{FML}}(K,v').
	\tag{1}
	\]
	For the reverse inclusion, suppose
$	\psi\notin\operatorname{Th}_{\mathrm{FML}}(M,u')$.
	Then
	$M,u'\models N\psi$.
	Hence $N\psi\in\Gamma$, and by \textup{(1)}
	$K,v'\models N\psi$.
	Thus $\psi$ is non-designated at $v'$, and so
$	\psi\notin\operatorname{Th}_{\mathrm{FML}}(K,v')$. Therefore $\operatorname{Th}_{\mathrm{FML}}(K,v')
\subseteq
\operatorname{Th}_{\mathrm{FML}}(M,u')$.
	Together with \textup{(1)}, this yields
$	\operatorname{Th}_{\mathrm{FML}}(M,u')
	=
	\operatorname{Th}_{\mathrm{FML}}(K,v')$,
	hence $u'Zv'$. This proves Forth. Back is symmetric, thus $Z$ is a bisimulation. Since the
	hypothesis implies $xZy$, we obtain
	$(M,x)\sim(K,y)$.

\end{proof}

\section{Saturated elementary extensions}
\label{sec:saturated}
We adapt the saturation construction of
Badia and Olkhovikov~\cite{badia2020lindstrom} to the present exact-test setting. The general
compactness--unravelling strategy is retained, while the exact
truth tests and the derived modality $\Diamond_1$ provide the
separation and successor-existence mechanisms required below.\\
For an abstract extension $\mathscr L$, a Kripke-model embedding
$f:M\to K$ is called $\mathscr L$-elementary if
$\operatorname{Th}_{\mathscr L}(M,x)
=
\operatorname{Th}_{\mathscr L}(K,f(x))$
for every $x\in W_M$. By Lemma~\ref{lem:recovery}, this is
equivalent to preservation of the exact value of every
$\theta\in\mathscr L(\tau)$.\\
Let $(M,w)$ be a pointed $\tau$-model. Its unravelling $M_w$ consists of all
finite $R_M$-paths starting at $w$, ordered by one-step
extension, with valuations inherited from endpoints; see, e.g.,
\cite{blackburn2001modal,badia2020lindstrom}.
The endpoint map $\pi:M_w\to M$ is a bisimulation.
\begin{prop}[Realization extension]
	\label{prop:realization}
	Let $\mathscr L$ be compact and strongly invariant under
	bisimulation. For every pointed $\tau$-model $(M,w)$, its unravelling
	$U=M_w$ admits an $\mathscr L$-elementary embedding $f:U\longrightarrow K$
	into an unravelled model $K$ such that every
	$\mathscr L$-successor type at $s\in U$ is realized by a
	successor of $f(s)$ in $K$.
\end{prop}
\begin{proof}
	Let $\tau$ be the signature of $M$, and let $r$ denote the root
	of the unravelling $U=M_w$. For $s\in U$, let $d(s)$ denotes its
	depth. For every $s\in U$, introduce a fresh propositional variable
	$P_s$. Moreover, for every $s\in U$ and every
	$\mathscr L$-successor type $\Gamma$ at $(U,s)$, introduce a
	fresh propositional variable $Q_{s,\Gamma}$. The variables $P_s$ mark the intended images of
	states, while $Q_{s,\Gamma}$ mark witnesses for successor types. Let $\tau^+$ be the
	resulting expansion of $\tau$. For $n<\omega$, put
	$\Box^0\varphi=\varphi,
	\;
	\Box^{n+1}\varphi=\Box(\Box^n\varphi)$, and define $\Diamond_1^n$ analogously. Let $\Sigma\subseteq\mathscr L(\tau^+)$ contain
	\[
	\Diamond_1^{\,d(s)}P_s,
	\quad s\in U,
	\tag{E1}
	\]
	\[
	\Box^n(T_0P_s\vee T_0P_t),
	\quad
	s\neq t,\quad d(s)=d(t)=n,
	\tag{E2}
	\]
	\[
	\Box^{d(s)}
	(T_0P_s\vee\Diamond_1P_t),
	\quad
	sR_Ut,
	\tag{E3}
	\]
	\[
	\Box^{d(s)}
	(T_0P_s\vee\theta),
	\quad
	\theta\in\operatorname{Th}_{\mathscr L}(U,s),
	\tag{E4}
	\]
	\[
	\Box^{d(s)}
	(T_0P_s\vee\Diamond_1Q_{s,\Gamma}),
	\tag{E5}
	\]
	for every $\mathscr L$-successor type $\Gamma$ at $(U,s)$, and
	\[
	\Box^{d(s)+1}
	(T_0Q_{s,\Gamma}\vee\gamma),
	\qquad
	\gamma\in\Gamma.
	\tag{E6}
	\]
	We first show that $\Sigma$ is finitely satisfiable. Let
	$\Sigma_0\subseteq\Sigma$ be finite. We expand the original model $U$ to a
	$\tau^+$-model $U^+$. For every $s\in U$, set
	\[
	V_{U^+}(P_s)(u)=
	\begin{cases}
		1,&u=s,\\
		0,&u\neq s.
	\end{cases}
	\tag{2}
	\]
	For each $Q_{s,\Gamma}$ occurring in $\Sigma_0$, let $\gamma_1,\ldots,\gamma_m$
	be all members of $\Gamma$ which occur in the corresponding
	instances of \textup{(E6)} contained in $\Sigma_0$. Since
	$\Gamma$ is an $\mathscr L$-successor type at $(U,s)$, there is
	some $u_{s,\Gamma}$ such that $sR_Uu_{s,\Gamma}$
	and $U,u_{s,\Gamma}\models_{\mathscr L}\gamma_i
	\quad (1\leq i\leq m)$. If $m=0$, the same conclusion that a successor of $s$ exists
	follows by applying the definition of successor type to the
	finite set $\varnothing$.	
	Put
	\[
	V_{U^+}(Q_{s,\Gamma})(u)=
	\begin{cases}
		1,&u=u_{s,\Gamma},\\
		0,&u\neq u_{s,\Gamma}.
	\end{cases}
	\tag{3}
	\]
	All remaining fresh variables may be interpreted constantly by
	$0$. By invariance under expansion, every
	$\theta\in\mathscr L(\tau)$ has the same value in $U$ and $U^+$. We now verify the relevant clauses. Formula \textup{(E1)} is true at
	the root since $s$ is reachable from $r$ in exactly $d(s)$
	steps and $P_s(s)=1$. For \textup{(E2)}, at a state of depth
	$n$ at most one of $P_s,P_t$ has value $1$; hence at least one
	of $T_0P_s,T_0P_t$ has value $1$. Clause \textup{(E3)} holds
	because, at the unique state $s$ where $P_s=1$, the actual
	successor $t$ witnesses $\Diamond_1P_t$; at all other states
	$T_0P_s=1$. Clause \textup{(E4)} holds because
	$\theta\in\operatorname{Th}_{\mathscr L}(U,s)$ at $s$, while
	$T_0P_s=1$ at all other states. Finally,
	\textup{(E5)}--\textup{(E6)} hold by the choice of
	$u_{s,\Gamma}$. Hence $U^+,r\models_{\mathscr L}\Sigma_0$.
	Thus every finite subset of $\Sigma$ is satisfiable. By
	compactness of $\mathscr L$, there is a pointed
	$\tau^+$-model $(B,b)$ such that
	\[
	B,b\models_{\mathscr L}\Sigma.
	\tag{4}
	\]
Let $B_0=B\upharpoonright\tau$ be the reduct to the original signature, and let
	$K=(B_0)_b$
	be its unravelling from $b$. Let $\pi:K\longrightarrow B_0$
	be the endpoint map. Since $\pi$ is a bisimulation, strong
	bisimulation invariance and reduct invariance give, for every $\theta\in\mathscr L(\tau)$
	and $u\in K$,
	=
	\[
	\llbracket\theta\rrbracket_K^{\mathscr L}(u)
	=\llbracket\theta\rrbracket_{B_0}^{\mathscr L}(\pi(u))
	=
	\llbracket\theta\rrbracket_B^{\mathscr L}(\pi(u)).
	\tag{5}
	\]
	We now construct
	$f:U\longrightarrow K$
	recursively on depth. Since $d(r)=0$, clause \textup{(E1)}
	gives $B,b\models P_r$. Set $f(r)=(b)$. Suppose $f(s)$ has been defined, has depth $d(s)$, and 
	\[
	B,\pi(f(s))\models P_s.
	\tag{6}
	\]
	Let $sR_Ut$. Since $f(s)$ has depth $d(s)$, clause
	\textup{(E3)}, evaluated at its endpoint, yields
	$B,\pi(f(s))\models\Diamond_1P_t$.
	By Lemma~\ref{lem:diamond}, there exists $c\in B$ such that
	$\pi(f(s))R_Bc
	\;\text{and}\;
	B,c\models P_t$. Define $f(t)$ to be the path obtained by appending $c$ to
	$f(s)$. Thus
	\[
sR_Ut
	\quad\Longrightarrow\quad
	f(s)R_Kf(t).
	\tag{7}
	\]
	We next show that $f$ is injective. Suppose $f(s)=f(t)$.
	Then $d(s)=d(t)=n$, and the common endpoint satisfies both
	$P_s=1$ and $P_t=1$. If $s\neq t$, clause \textup{(E2)} gives
	at this endpoint $T_0P_s\vee T_0P_t=1$.
	But $P_s=P_t=1$ implies $T_0P_s=T_0P_t=0$,
	a contradiction. Hence $s=t$, and $f$ is injective. The accessibility relation is also reflected. Suppose
	$f(s)R_Kf(t)$.
	Then $t$ is non-root; let $p(t)$ denote its unique immediate predecessor in
	the tree $U$. By the recursive construction, $f(p(t))R_Kf(t)$.
	Since every non-root node of the unravelling $K$ has a unique
	immediate predecessor,
	$f(s)=f(p(t))$.
	Injectivity of $f$ gives $s=p(t)$,
	and therefore $sR_Ut$.
	Together with \textup{(7)}, we have
	\[
	sR_Ut
	\quad\Longleftrightarrow\quad
	f(s)R_Kf(t).
	\tag{8}
	\]
	It remains to prove $\mathscr L$-elementarity. Let
	$\theta\in\operatorname{Th}_{\mathscr L}(U,s)$.
	By clause \textup{(E4)}, since the endpoint of $f(s)$ has
	depth $d(s)$ and satisfies $P_s=1$,
	$B,\pi(f(s))\models_{\mathscr L}\theta$.
	Using \textup{(5)},
	$K,f(s)\models_{\mathscr L}\theta$.
	Therefore
	\[
	\operatorname{Th}_{\mathscr L}(U,s)
	\subseteq
	\operatorname{Th}_{\mathscr L}(K,f(s)).
	\tag{9}
	\]
	For the converse, let $\theta\notin\operatorname{Th}_{\mathscr L}(U,s)$.
	Then
	$\llbracket\theta\rrbracket_U^{\mathscr L}(s)\neq1$.
	Hence the non-designation test $\mathsf N\theta:=T_0(T_1\theta)$
	has value $1$ at $s$, and consequently
	$\mathsf N\theta
	\in
	\operatorname{Th}_{\mathscr L}(U,s)$.
	By \textup{(9)}, $K,f(s)\models_{\mathscr L}\mathsf N\theta$.
	By the semantics of $\mathsf N$, $\llbracket\theta\rrbracket_K^{\mathscr L}(f(s))\neq1$,
	so $\theta\notin
	\operatorname{Th}_{\mathscr L}(K,f(s))$.
	Thus the reverse inclusion holds, and 
	\[
	\operatorname{Th}_{\mathscr L}(U,s)
	=
	\operatorname{Th}_{\mathscr L}(K,f(s)).
	\tag{10}
	\]
	By Lemma \ref{lem:recovery}, every
	$\theta\in\mathscr L(\tau)$ has the same exact value at $s$ and
	$f(s)$. In particular, atomic valuations are preserved. Together
	with injectivity and \textup{(8)}, this shows that $f$ is an
	$\mathscr L$-elementary Kripke-model embedding. Finally, let $\Gamma$ be an $\mathscr L$-successor type at
	$(U,s)$. Clause \textup{(E5)}, evaluated at the endpoint of
	$f(s)$, gives $B,\pi(f(s))\models
	\Diamond_1Q_{s,\Gamma}$.
	Hence there exists $c$ with $\pi(f(s))R_Bc
	\;\text{and}\;
	B,c\models Q_{s,\Gamma}$.
	Let $u\in K$ be the one-step extension of $f(s)$ obtained by
	appending $c$. Then $f(s)R_Ku$.
	Let $\gamma\in\Gamma$. Since $c$ is reached from $b$ in
	$d(s)+1$ steps, clause \textup{(E6)} gives $B,c\models_{\mathscr L}
	T_0Q_{s,\Gamma}\vee\gamma$.
	Since $Q_{s,\Gamma}$ has value $1$ at $c$, $T_0Q_{s,\Gamma}$
	has value $0$ there. Hence $B,c\models_{\mathscr L}\gamma$.
	Since $\gamma\in\mathscr L(\tau)$, equation
	\textup{(5)} yields $K,u\models_{\mathscr L}\gamma$.
	Since $\gamma\in\Gamma$ was arbitrary, $K,u\models_{\mathscr L}\Gamma$. As $f(s)R_Ku$, the type $\Gamma$ is realized by a successor of $f(s)$ in $K$.
	This completes the proof.
\end{proof}
\begin{prop}[Saturated elementary extension]
	\label{prop:saturated}
	Suppose that $\mathscr L$ is compact, has the Tarski Union
	Property, and is strongly invariant under bisimulation.
	Then every unravelling $M_w$ has an
	$\mathscr L$-elementary $\mathscr L$-saturated extension.
\end{prop}
\begin{proof}
	Put $M_0:=M_w$.
	Since the model obtained in
	Proposition~\ref{prop:realization} is again an unravelling, we
	may iterate that construction. After identifying each model with
	its $\mathscr L$-elementary image, we obtain a chain
	\[
	M_0
	\preccurlyeq_{\mathscr L}
	M_1
	\preccurlyeq_{\mathscr L}
	M_2
	\preccurlyeq_{\mathscr L}
	\cdots
	\tag{11}
	\]
	such that, for every $n<\omega$, every
	$\mathscr L$-successor type at a state of $M_n$ is realized
	by a successor in $M_{n+1}$. Let $M_\omega:=\bigcup_{n<\omega}M_n$.
	By the Tarski Union Property,
	\[
	M_n\preccurlyeq_{\mathscr L}M_\omega
	\qquad
	\text{for every }n<\omega.
	\tag{12}
	\]
	We show that $M_\omega$ is $\mathscr L$-saturated.
	Let $x\in W_{M_\omega}$, and let $\Gamma\subseteq\mathscr L(\tau)$
	be an $\mathscr L$-successor type at $(M_\omega,x)$.
	Choose $n<\omega$ such that $x\in W_{M_n}$.
	We first show that $\Gamma$ is an
	$\mathscr L$-successor type at $(M_n,x)$.
	Let $\Gamma_0\subseteq\Gamma$
	be finite. Since $\Gamma$ is a successor type at
	$(M_\omega,x)$, Lemma~\ref{lem:diamond} gives
	\[
	M_\omega,x\models_{\mathscr L}
	\Diamond_1\!\left(\bigwedge\Gamma_0\right).
	\tag{13}
	\]
	Since $M_n\preccurlyeq_{\mathscr L}M_\omega$, elementarity gives
	\[
	M_n,x\models_{\mathscr L}
	\Diamond_1\!\left(\bigwedge\Gamma_0\right).
	\tag{14}
	\]
	Applying Lemma~\ref{lem:diamond} again, there exists
	$y_{\Gamma_0}\in W_{M_n}$ such that $xR_{M_n}y_{\Gamma_0}$
	and $M_n,y_{\Gamma_0}\models_{\mathscr L}\Gamma_0$.
	Since this holds for every finite
	$\Gamma_0\subseteq\Gamma$, the set $\Gamma$ is an
	$\mathscr L$-successor type at $(M_n,x)$. By the construction of $M_{n+1}$, there exists
	$z\in W_{M_{n+1}}$ such that $xR_{M_{n+1}}z$
	and
	\[
	M_{n+1},z\models_{\mathscr L}\Gamma.
	\tag{15}
	\]
	Now, by \textup{(12)}, $M_{n+1}\preccurlyeq_{\mathscr L}M_\omega$; hence $M_\omega,z\models_{\mathscr L}\Gamma$. Moreover, $R_{M_\omega}=\bigcup_{n<\omega}R_{M_n}$ so, $xR_{M_{n+1}}z\;\Longrightarrow\; xR_{M_\omega}z$. Thus $\Gamma$ is realized by a successor of $x$ in $M_\omega$.
	Therefore $M_\omega$ is $\mathscr L$-saturated. Finally, $M_w=M_0\preccurlyeq_{\mathscr L}M_\omega$ by \textup{(12)}, so $M_\omega$ is an
	$\mathscr L$-elementary $\mathscr L$-saturated extension of $M_w$.

\end{proof}
\begin{note}
Since $\mathsf{FML}_{\mathcal L}(\tau)\subseteq\mathscr L(\tau)$,
every $\mathsf{FML}_{\mathcal L}$-successor type is also an
$\mathscr L$-successor type. Hence every $\mathscr L$-saturated
model is, in particular,
$\mathsf{FML}_{\mathcal L}$-saturated.
\end{note}
\section{Lindstr\"om maximality}
\label{sec:lindstrom}
The proof of maximality requires two pointed models which agree on
all $\mathsf{FML}_{\mathcal L}$-formulas but disagree on a
hypothetical new formula of the abstract extension. In the finite
MTL-chain argument of Badia and Olkhovikov~\cite{badia2020lindstrom},
the corresponding separation step uses the unique immediate
predecessor of $1$. For a finite Heyting algebra that is not
necessarily a chain, there need not be a greatest proper element
below $1$. We therefore replace this chain-dependent device by the
Boolean-valued pair $D\psi=T_1\psi,\quad N\psi=T_0(T_1\psi)$.
\begin{lem}[Decision lemma]
	\label{lem:decision}
Let $\chi\in\mathscr L(\tau)$ be not $1$-equivalent to any
formula of $\mathsf{FML}_{\mathcal L}(\tau)$, and let $\psi\in\mathsf{FML}_{\mathcal L}(\tau)$.
Then at least one of $\chi\wedge D\psi,
\;
\chi\wedge N\psi$
is not $1$-equivalent to any formula of
$\mathsf{FML}_{\mathcal L}(\tau)$.
\end{lem}
\begin{proof}
	Suppose, towards a contradiction, that there are $\theta_1,\theta_0\in\mathsf{FML}_{\mathcal L}(\tau)$
	such that
	\[
	\chi\wedge D\psi\equiv_1\theta_1,
	\qquad
	\chi\wedge N\psi\equiv_1\theta_0.
	\tag{16}
	\]
	We claim that
	\[
	\chi\equiv_1D\theta_1\vee D\theta_0.
	\tag{17}
	\]
	Let $(M,x)$ be an arbitrary pointed $\tau$-model. Suppose first that $\llbracket\chi\rrbracket_M^{\mathscr L}(x)=1$. By the definitions of $D$ and $N$, $\bigl(
	\llbracket D\psi\rrbracket_M^{\mathscr L}(x),
	\llbracket N\psi\rrbracket_M^{\mathscr L}(x)
	\bigr)
	\in\{(1,0),(0,1)\}$. Hence $\bigl(
	\llbracket\chi\wedge D\psi\rrbracket_M^{\mathscr L}(x),
	\llbracket\chi\wedge N\psi\rrbracket_M^{\mathscr L}(x)
	\bigr)
	\in\{(1,0),(0,1)\}$. By \textup{(16)}, it follows that $\llbracket\theta_1\rrbracket_M^{\mathscr L}(x)=1
	\;\text{or}\;
	\llbracket\theta_0\rrbracket_M^{\mathscr L}(x)=1$.
	Therefore
	$\llbracket D\theta_1\vee D\theta_0
	\rrbracket_M^{\mathscr L}(x)=1$. Conversely, suppose that
	$\llbracket D\theta_1\vee D\theta_0
	\rrbracket_M^{\mathscr L}(x)=1$.
	Since $D\theta_1$ and $D\theta_0$ are Boolean-valued,
	$\llbracket D\theta_1\rrbracket_M^{\mathscr L}(x)=1
	\;\text{or}\;
	\llbracket D\theta_0\rrbracket_M^{\mathscr L}(x)=1$.
	Hence
	$\llbracket\theta_1\rrbracket_M^{\mathscr L}(x)=1
	\;\text{or}\;
	\llbracket\theta_0\rrbracket_M^{\mathscr L}(x)=1$.
	By \textup{(16)}, respectively,
	$\llbracket\chi\wedge D\psi\rrbracket_M^{\mathscr L}(x)=1
	\;\text{or}\;
	\llbracket\chi\wedge N\psi\rrbracket_M^{\mathscr L}(x)=1$.
	In either case,
	$\llbracket\chi\rrbracket_M^{\mathscr L}(x)=1$.
	Thus \textup{(17)} holds. Since
	$D\theta_1\vee D\theta_0\in\mathsf{FML}_{\mathcal L}(\tau)$, this contradicts the assumption on $\chi$.
\end{proof}
\begin{thm}[Lindstr\"om maximality]
	\label{thm:lindstrom}
	Let $\mathcal L$ be a fixed finite Heyting algebra, and let
	$\mathscr L$ be an abstract extension of
	$\mathsf{FML}_{\mathcal L}$. If $\mathscr L$ is compact, has the
	Tarski Union Property, and is strongly invariant under
	bisimulation, then $\mathscr L\equiv_1\mathsf{FML}_{\mathcal L}$.
	Consequently, $\mathsf{FML}_{\mathcal L}$ is maximal with
	respect to $1$-expressive power among its abstract extensions
	satisfying these three properties.
\end{thm}
\begin{proof}
	Suppose, towards a contradiction, that $\mathscr L\not\leq_1\mathsf{FML}_{\mathcal L}$.
	Then there is a formula $\varphi$ of $\mathscr L$ which is not
	$1$-equivalent to any formula of
	$\mathsf{FML}_{\mathcal L}$. By the finite occurrence property, there is a finite signature
	$\tau$ such that $\varphi\in\mathscr L(\tau)$. Since both
	$\tau$ and $\mathcal L$ are finite and every formula is a finite
	expression, $\mathsf{FML}_{\mathcal L}(\tau)$ is countable. Fix an enumeration $\psi_1,\psi_2,\ldots$
	of all its formulas. Put $F_0:=1$.
	We recursively choose $\delta_n\in\{D\psi_n,N\psi_n\}$,
	and define $F_n:=\bigwedge_{i=1}^{n}\delta_i$,
	so that
	\[
	\varphi\wedge F_n
	\tag{18}
	\]
	is not $1$-equivalent to any formula of
	$\mathsf{FML}_{\mathcal L}(\tau)$. For $n=0$ this follows from
	the choice of $\varphi$. Suppose that $\delta_1,\ldots,\delta_{n-1}$ have 
	been chosen and that $\varphi\wedge F_{n-1}$
	is not $1$-equivalent to any
	$\mathsf{FML}_{\mathcal L}(\tau)$-formula. Applying Lemma~\ref{lem:decision} to $\chi=\varphi\wedge F_{n-1}
	\;\text{and}\;
	\psi=\psi_n$,
	at least one of $(\varphi\wedge F_{n-1})\wedge D\psi_n,
	\;
	(\varphi\wedge F_{n-1})\wedge N\psi_n$
	is not $1$-equivalent to any formula of
	$\mathsf{FML}_{\mathcal L}(\tau)$. Choose $\delta_n$ accordingly.
	This completes the recursion. We next establish two satisfiability facts. First,
	\[
	\varphi\wedge F_n
	\tag{19}
	\]
	is satisfiable for every $n\geq 0$. Otherwise $\llbracket\varphi\wedge F_n\rrbracket_M^{\mathscr L}(x)\neq1$ for every pointed $\tau$-model $(M,x)$, and hence $\varphi\wedge F_n\equiv_1 0$,
	contrary to \textup{(18)}. Second, we claim that
	\[
	N\varphi\wedge F_n
	\tag{20}
	\]
	is also satisfiable for every $n\geq 0$. Suppose otherwise.
	Then, for every pointed $\tau$-model $(K,x)$, $K,x\models F_n
	\;\Longrightarrow\;
	K,x\not\models_{\mathscr L}N\varphi$. Since $N\varphi=T_0(T_1\varphi)$ is Boolean-valued,
	$K,x\models F_n
	\;\Longrightarrow\;
	\llbracket N\varphi\rrbracket_K^{\mathscr L}(x)=0
	\;\Longrightarrow\;
	\llbracket\varphi\rrbracket_K^{\mathscr L}(x)=1$.
	Therefore
	$K,x\models F_n
	\;\Longrightarrow\;
	K,x\models_{\mathscr L}\varphi\wedge F_n$.

	The converse implication is immediate, so $\varphi\wedge F_n\equiv_1 F_n$. Since $F_n\in\mathsf{FML}_{\mathcal L}(\tau)$,
	this contradicts \textup{(18)}. Thus \textup{(20)} is satisfiable. Now define $\Gamma^+
	=
	\{\varphi\}\cup\{\delta_n:n\geq1\}$
	and $\Gamma^-
	=
	\{N\varphi\}\cup\{\delta_n:n\geq1\}$.
	Both sets are finitely satisfiable. Indeed, let
	$\Delta\subseteq\Gamma^+$ be finite, and choose $m$ at least as
	large as every index of a decision formula occurring in
	$\Delta$ (taking $m=0$ if none occurs). A pointed model satisfying $\varphi\wedge F_m$
	also satisfies $\Delta$. Similarly, satisfiability of $N\varphi\wedge F_m$,
	shows that every finite subset of $\Gamma^-$ is satisfiable. By compactness of $\mathscr L$, there are pointed $\tau$-models
	$(M,w)$ and $(K,v)$ such that $M,w\models_{\mathscr L}\Gamma^+,
	\;
	K,v\models_{\mathscr L}\Gamma^-$.
	Consequently,
	\[
	\llbracket\varphi\rrbracket_M^{\mathscr L}(w)=1,
	\qquad
	\llbracket\varphi\rrbracket_K^{\mathscr L}(v)\neq1.
	\tag{21}
	\]
	We now show that, for every
	$\psi\in\mathsf{FML}_{\mathcal L}(\tau)$, $M,w\models\psi
	\Longleftrightarrow
	K,v\models\psi$.
	
	Fix $n\geq1$. Since $\delta_n\in\{D\psi_n, N\psi_n\}$, and both $(M,w)$ and $(K,v)$ satisfy $\delta_n$, there are two cases. If$\delta_n=D\psi_n$, then 
$	M,w\models\psi_n
	\; \text{and}\;
	K,v\models\psi_n$. If $\delta_n=N\psi_n$, then 
	$M,w\not\models\psi_n
	\; \text{and}\;
	K,v\not\models\psi_n$. Therefore, 
	$M,w\models\psi_n
	\Longleftrightarrow
	K,v\models\psi_n$, for every $n\geq1$.

	
	 Hence
	\[
	\operatorname{Th}_{\mathsf{FML}_{\mathcal L}}(M,w)
	=
	\operatorname{Th}_{\mathsf{FML}_{\mathcal L}}(K,v).
	\tag{22}
	\]
	By exact-value recovery, Lemma \ref{lem:recovery}) applied to
	$\mathsf{FML}_{\mathcal L}$ yields
	\[
	\llbracket\psi\rrbracket_M(w)
	=
	\llbracket\psi\rrbracket_K(v)
	\tag{23}
	\]
	for every $\psi\in\mathsf{FML}_{\mathcal L}(\tau)$. Let $U=M_w,
	\;
	V=K_v$
	be the unravellings of the two pointed models, with roots
	$r_U$ and $r_V$, respectively. Their endpoint maps are
	bisimulations. Therefore strong bisimulation invariance of
	$\mathscr L$ preserves the disagreement in \textup{(21)}, so
	\[
	\llbracket\varphi\rrbracket_U^{\mathscr L}(r_U)=1,
	\qquad
	\llbracket\varphi\rrbracket_V^{\mathscr L}(r_V)\neq1.
	\tag{24}
	\]
	By Proposition \ref{prop:base-bisim} and \textup{(23)}, 
	\[
	\llbracket\psi\rrbracket_U(r_U)
	=
	\llbracket\psi\rrbracket_V(r_V)
	\tag{25}
	\]
	for every $\psi\in\mathsf{FML}_{\mathcal L}(\tau)$. By Proposition~\ref{prop:saturated}, take
	$\mathscr L$-elementary $\mathscr L$-saturated extensions $U\preccurlyeq_{\mathscr L}\widehat M,
	\;
	V\preccurlyeq_{\mathscr L}\widehat K$.
	Let $\widehat r_U$ and $\widehat r_V$ denote the images of the
	two roots. Since these embeddings are $\mathscr L$-elementary,
	Lemma~\ref{lem:recovery} gives preservation of the exact values of all
	$\mathscr L(\tau)$-formulas. Hence \textup{(24)} becomes
	\[
	\llbracket\varphi\rrbracket_{\widehat M}^{\mathscr L}
	(\widehat r_U)=1,
	\qquad
	\llbracket\varphi\rrbracket_{\widehat K}^{\mathscr L}
	(\widehat r_V)\neq1,
	\tag{26}
	\]
	while \textup{(25)} yields
	\[
	\llbracket\psi\rrbracket_{\widehat M}(\widehat r_U)
	=
	\llbracket\psi\rrbracket_{\widehat K}(\widehat r_V)
	\tag{27}
	\]
	for every
	$\psi\in\mathsf{FML}_{\mathcal L}(\tau)$. Since every $\mathscr L$-saturated model is, in particular,
	$\mathsf{FML}_{\mathcal L}$-saturated,
	Proposition~\ref{prop:hm} applied to \textup{(27)} gives $(\widehat M,\widehat r_U)
	\sim
	(\widehat K,\widehat r_V)$.
	Strong bisimulation invariance of $\mathscr L$ therefore implies $\llbracket\varphi\rrbracket_{\widehat M}^{\mathscr L}
	(\widehat r_U)
	=
	\llbracket\varphi\rrbracket_{\widehat K}^{\mathscr L}
	(\widehat r_V)$,
	contradicting \textup{(26)}. 
	Hence $\mathscr L\leq_1\mathsf{FML}_{\mathcal L}$.Since
	$\mathsf{FML}_{\mathcal L}\subseteq\mathscr L$, the reverse $1$-expressive inequality
	$\mathsf{FML}_{\mathcal L}\leq_1\mathscr L$
	is immediate.
	Hence $\mathscr L\equiv_1\mathsf{FML}_{\mathcal L}$.
\end{proof}

\begin{cor}[Exact-value definability]
	\label{cor:value-fibres}
Under the assumptions of Theorem~\ref{thm:lindstrom}, for every $\varphi\in\mathscr L(\tau)
\;\text{and}\;
\ell\in L$,
there exists a formula $\theta_{\varphi,\ell}
\in\mathsf{FML}_{\mathcal L}(\tau)$
such that, for every pointed $\tau$-model $(M,x)$, $M,x\models\theta_{\varphi,\ell}
\Longleftrightarrow
\llbracket\varphi\rrbracket_M^{\mathscr L}(x)=\ell$.
\end{cor}

\begin{proof}
	Since $\mathscr L$ is closed under exact truth tests, $T_\ell\varphi\in\mathscr L(\tau)$.
	By Theorem~\ref{thm:lindstrom}, there exists $\theta_{\varphi,\ell}
	\in\mathsf{FML}_{\mathcal L}(\tau)$
	such that $T_\ell\varphi\equiv_1\theta_{\varphi,\ell}$.
	Hence, for every pointed $\tau$-model $(M,x)$,
		$M,x\models\theta_{\varphi,\ell}
		\Longleftrightarrow
		M,x\models_{\mathscr L}T_\ell\varphi
		\Longleftrightarrow
		\llbracket T_\ell\varphi\rrbracket_M^{\mathscr L}(x)=1
		\Longleftrightarrow
		\llbracket\varphi\rrbracket_M^{\mathscr L}(x)=\ell$,
	where the last equivalence follows from the definition of
	$T_\ell$.
\end{proof}
\begin{rem}
	Corollary~\ref{cor:value-fibres} provides, separately for each
	$\ell\in L$, an $\mathsf{FML}_{\mathcal L}(\tau)$-definition of the
	exact-value fibre $\{(M,x):
	\llbracket\varphi\rrbracket_M^{\mathscr L}(x)=\ell\}$.
	It does not assert the existence of a single
	$\theta_\varphi\in\mathsf{FML}_{\mathcal L}(\tau)$ such that $\llbracket\theta_\varphi\rrbracket_M(x)
	=
	\llbracket\varphi\rrbracket_M^{\mathscr L}(x)$
	for every pointed $\tau$-model $(M,x)$.
\end{rem}

\section{Conclusion and future work}
\label{sec:conclusion}
We have established a Lindstr\"om-style maximality theorem for
Maruyama's exact-truth-test presentation of Fitting's modal logic
over a fixed finite Heyting algebra $\mathcal L$, without assuming
linearity. This extends the finite-chain result of Badia and
Olkhovikov~\cite{badia2020lindstrom}: the Boolean-valued tests $D\varphi=T_1\varphi,
\;N\varphi=T_0(T_1\varphi)$
replace the chain-dependent separation mechanism and support the
saturation and maximality arguments. Thus
$\mathsf{FML}_{\mathcal L}$ is maximal with respect to
$1$-expressive power among its compact abstract extensions having
the Tarski Union Property and strong bisimulation invariance.
Moreover, Corollary~\ref{cor:value-fibres} shows that each
exact-value fibre of an extension formula is definable by an
$\mathsf{FML}_{\mathcal L}$-formula.

Several questions remain open. A first question is whether the
Tarski Union Property can be replaced by a suitable
relativization or finite-depth condition, as in other modal
Lindstr\"om characterizations
\cite{deRijke1995,badia2020lindstrom}. A second direction is to
replace crisp accessibility by genuinely $\mathcal L$-valued
accessibility; this would require a different analysis of
bisimulation, successor types, and the derived existential
construction.\\
A further question concerns the precise algebraic role of the
exact truth tests: which finite truth-value algebras, or which
weaker families of value-separating operations, support the
Boolean separation and exact-value recovery required by the
maximality argument? Such a characterization would isolate the
algebraic content of the proof from the particular
exact-truth-test presentation used here. It would also be interesting to investigate whether the present exact-test maximality argument admits a
coalgebraic formulation, in connection with many-valued coalgebraic logics over semi-primal varieties~\cite{kurz2024many}.
Finally, it remains open whether the separate exact-value definitions supplied by Corollary~\ref{cor:value-fibres} can be replaced by a single
$\mathsf{FML}_{\mathcal L}$-formula reproducing the full $\mathcal L$-valued semantics of an arbitrary extension formula.

\section*{Funding}
This research received no external funding.
\section*{Conflict of Interest Statement}
The authors declare that they have no known competing financial interests or personal relationships that could have appeared to influence the work reported in this paper.
\section*{Data Availability Statement}
Data sharing is not applicable to this article as no new data were created or analyzed in this study.

\end{document}